\documentclass[a4paper,envcountsame]{llncs}

\usepackage{amsfonts}
\usepackage{amsmath}
\usepackage{amssymb}
\usepackage{tikz}

\usetikzlibrary{automata}

\DeclareMathOperator{\enc}{Enc}
\DeclareMathOperator{\plays}{Plays}

\newcommand{\nats}{\mathbb N}
\newcommand{\ints}{\mathbb Z}

\newcommand{\s}{\mathsf s}
\newcommand{\p}{\mathcal P}
\newcommand{\A}{\mathcal A}

\begin{document}

\title{Reachability in Two-Clock Timed Automata is PSPACE-complete}
\author{John Fearnley \inst{1} \and Marcin Jurdzi\'nski\inst{2} }
\institute{
Department of Computer Science, University of Liverpool, UK \and
Department of Computer Science, University of Warwick, UK 
}
\maketitle

\begin{abstract}
A recent result has shown that reachability in two-clock timed automata is
log-space equivalent to reachability in bounded one-counter
automata~\cite{HOW12}. We show that reachability in bounded one-counter automata
is PSPACE-complete. 
\end{abstract}

\section{Introduction}

Timed automata~\cite{AD94} are a successful and widely used formalism, which are
used in the analysis and verification of real time systems. A timed automaton is
a non-deterministic finite automaton that is equipped with a number of
real-valued \emph{clocks}, which allow the automaton to measure the passage of
time. 

Perhaps the most fundamental problem for timed automata is the
\emph{reachability} problem: given an initial state, can the automaton perform a
sequence of transitions in order to reach a specified target state? In their
seminal paper on timed automata~\cite{AD94}, Alur and Dill showed that this
problem is PSPACE-complete. To show hardness for PSPACE, their proof starts with
a linear bounded automaton (LBA), which is a non-deterministic Turing machine
with a fixed tape length~$n$. They produce a timed automaton with $2n + 1$
clocks, and showed that the timed automaton can reach a specified state if and
only if the LBA halts.

However, the work of Alur and Dill did not address the case where the number of
clocks is small. This was rectified by Courcoubetis and Yannakakis~\cite{CY92},
who showed that reachability in timed automata with only three clocks is still
PSPACE-complete. Their proof cleverly encodes the tape of an LBA in a single
clock, and then uses the two additional clocks to perform
all necessary operations on the encoded tape. In contrast to this, Laroussinie
et al.\ have shown that reachability in one-clock timed automata is complete for
NLOGSPACE, and therefore no more difficult than computing reachability in
directed graphs~\cite{LMS04}. 

The complexity of reachability in two-clock timed automata has been left
open.  The best known lower bound was given by Laroussinie et al., who gave a
proof that the problem is NP-hard via a very natural reduction from 
subset-sum~\cite{LMS04}. Moreover, the problem lies in PSPACE, because
reachability in two-clock timed automata is obviously easier than reachability
in three-clock timed automata. However, the PSPACE-hardness proof of
Courcoubetis and Yannakakis seems to fundamentally require three clocks, and
does not naturally extend to the two-clock case. 
Naves~\cite{Naves} has shown that several extensions to two-clock timed
automata lead to PSPACE-completeness, but his work does not advance upon the
NP-hard lower bound for unextended two-clock timed automata.

In a recent paper, Haase et al.\ have shown a link between reachability in timed
automata and reachability in \emph{bounded counter automata}~\cite{HOW12}. A
bounded counter automaton is a non-deterministic finite automaton equipped with
a set of counters, and the transitions of the automaton may add or subtract
arbitrary integer constants to the counters. The
state space of each counter is bounded by some natural number $b$, so the
counter may only take values in the range $[0, b]$. Moreover, transitions may
only be taken if they do not increase or decrease a counter beyond the allowable
bounds. This gives these seemingly simple automata a surprising amount of power,
because the bounds can be used to implement inequality tests against the
counters.

Haase et al.\ show that reachability in two-clock timed automata is
log-space equivalent to reachability in bounded \emph{one}-counter
automata. Reachability in bounded one-counter automata has also been studied in
the context of one-clock timed automata with energy constraints~\cite{PFGMS08},
where it was shown that the problem lies in PSPACE, and is NP-hard. It has also
been shown that the reachability problem for \emph{unbounded} one-counter
automata is NP-complete~\cite{HKOW09}, but the NP containment proof does not
seem to generalise to bounded one-counter automata.

\paragraph{\bf Our contribution.} 

We show that satisfiability for quantified boolean formulas can be reduced, in
polynomial time, to reachability in bounded one-counter automata. Hence, we show
that reachability in bounded one-counter automata is PSPACE-complete, and
therefore we resolve the complexity of reachability in two-clock timed automata.
Our reduction uses two intermediate steps: \emph{subset-sum games} and
\emph{bounded counter-stack automata}.

Counter automata are naturally suited for solving subset-sum problems, so our
reduction starts with a quantified version of subset-sum, which we call
subset-sum games. One interpretation of satisfiability for quantified boolean
formulas is to view the problem as a game between an \emph{existential} player,
and a \emph{universal} player. The players take in turns to set their
propositions to true or false, and the existential
player wins if and only if the boolean formula is satisfied. Subset-sum games
follow the same pattern, but apply it to subset-sum: the two players alternate
in choosing numbers from sets, and the existential player wins if and
only if the chosen numbers sum to a given target. Previous work by Travers can
be applied to show that subset-sum games are PSPACE-complete~\cite{T11}.

We reduce subset-sum games to reachability in bounded one-counter automata.
However, we will not do this directly. Instead, we introduce bounded
counter-stack automata, which are able to store multiple counters, but have a
stack-like restriction on how these counters may be accessed. These automata are
a convenient intermediate step, because having access to multiple counters makes
it easier for us to implement subset-sum games. Moreover, the stack based
restrictions means that it is relatively straightforward to to show that
reachability in bounded counter-stack automata is reducible, in polynomial time,
to reachability in bounded one-counter automata, which completes our result.

\section{Bounded one-counter automata} 

A bounded one-counter automaton has a single counter that can store
values between $0$ and some bound $b \in \nats$. The automaton may add or
subtract values from the counter, so long as the bounds of $0$ and $b$ are not
overstepped. This can be used to test inequalities against the
counter. For example, to test whether the counter is larger than some $n \in
\nats$, we first attempt to subtract $n + 1$ from the counter, then, if that
works, we add $n+1$ back to the counter. This creates a sequence of two
transitions which can be taken if, and only if, the counter is greater than $n$.
A similar construction can be given for less-than tests. For the sake of
convenience, we will include explicit inequality testing in our formal
definition, with the understanding that this is not actually necessary.

We now give a formal definition. For two integers $a, b \in \ints$ we define
$[a, b] = \{n \in \ints \; : \; a \le n \le b\}$ to be the subset of integers
between~$a$ and~$b$. A bounded one-counter automaton is defined by a tuple $(L,
b, \Delta, l_0)$, where $L$ is a finite set of locations, $b \in \nats$ is a
global counter bound, $\Delta$ specifies the set of transitions,  and $l_0 \in
L$ is the initial location. Each transition in $\Delta$ has the form $(l, p,
g_1, g_2, l')$, where $l$ and $l'$ are locations, $p \in [-b, b]$ specifies how
the counter should be modified, and $g_1, g_2 \in [0, b]$ give lower and upper
guards for the counter.

Each state of the automaton consists of a location $l \in L$ along with a
counter value $c$. Thus, we define the set of states to be $L \times [0, b]$. A
transition exists between a state $(l, c) \in S$, and a state $(l', c') \in S$
if there is a transition $(l, p, g_1, g_2, l') \in \Delta$, where $g_1 \le c \le
g_2$, and $c' = c + p$.

The reachability problem for bounded one-counter automaton is: starting at the
state $(l_0, 0)$, can the automaton reach a specified target state $(l_t, c_t)$?
It has been shown that the reachability problem for bounded one-counter automata
is equivalent to the reachability problem for two-clock timed automata.

\begin{theorem}[\cite{HOW12}]
Reachability in bounded one-counter automata is log-space equivalent to
reachability in two-clock timed automata.
\end{theorem}

\section{Subset-sum games}

A subset-sum game is played between an \emph{existential} player and a
\emph{universal} player. The game is specified by a pair $(\psi, T)$,  where $T
\in \nats$, and $\psi$ is a list:
\begin{equation*}
\forall \; \{A_1, B_1\} \; \exists \; \{E_1, F_1\} \; \dots \; \forall \; \{A_n,
B_n\} \; \exists \; \{E_n, F_n\},
\end{equation*} 
where $A_i, B_i, E_i$, and $F_i$, are all natural numbers.
 
The game is played in rounds. In the first round, the universal player chooses
an element from $\{A_1, B_1\}$, and the existential player responds by choosing
an element from $\{E_1, F_1\}$. In the second round, the universal player
chooses an element from $\{A_2, B_2\}$, and existential player responds by
choosing an element from $\{E_2, F_2\}$. This pattern repeats for rounds $3$
through~$n$. Thus, at the end of the game, the players will have constructed a
sequence of numbers, and the existential player wins if and only if the sum of
these numbers is~$T$. 

Formally, the set of \emph{plays} of the game is the set:
\begin{equation*}
\p = \prod_{1 \le j \le n} \{A_j, B_j\} \times \{E_j, F_j\}.
\end{equation*}
A play~$P \in \p$ is winning for the existential player if and only if $\sum P = T$.

A strategy for the existential player is a list of functions $\s = (s_1, s_2,
\dots, s_n)$, where each function $s_i$ dictates how the existential player
should play in the $i$th round of the game. Thus, each function $s_i$ is of the
form:
\begin{equation*}
s_i : \prod_{1 \le j \le i} \{A_j, B_j\} \rightarrow \{E_i, F_i\}.
\end{equation*}
This means that the function $s_i$ maps the first~$i$ moves of the universal
player to a decision for the existential player in the $i$th round.

A play~$P$ conforms to a strategy~$\s$ if the decisions made by the existential
player in~$P$ always agree with~$\s$. More formally, for each $i$ in the range
$1 \le i \le n$, we define $F_i = P \cap \prod_{1 \le j \le i} \{A_j, B_j\}$ to
be the first $i$ moves made by the universal player. The play $P$ conforms to a
strategy $\s = (s_1, s_2, \dots, s_k)$ if $s_i(F_i) \in P$, for all $i$.
Given a strategy $\s$, we define the set of conforming plays to be
$\plays(\s)$.
Note that, since the universal player makes exactly $n$ choices, the set
$\plays(\s)$ contains exactly $2^n$ different plays.

A strategy $\s$ is \emph{winning} if every play $P \in \plays(\s)$ is
winning for the existential player. The \emph{subset-sum game problem} is to
decide, for a given SSG instance $(\psi, T)$, whether the existential player has
a winning strategy for $(\psi, T)$.

The SSG problem clearly lies in PSPACE, because it can be solved on a
polynomial time alternating Turing machine. A quantified version of subset-sum
has been shown to be PSPACE-hard, via a reduction from quantified boolean
formulas~\cite{T11}. Since SSGs are essentially a quantified version of
subset-sum, the proof of PSPACE-hardness easily carries over. See
Appendix~\ref{app:qss} for further details.

\begin{lemma}
\label{lem:qss}
The subset-sum game problem is PSPACE-complete.
\end{lemma}

\section{Counter-Stack Automata}
\label{sec:stack}

\paragraph{\bf Outline.} 

In this section we ask: can we use a bounded one-counter automaton to store
multiple counters? The answer is yes, but doing so forces an interesting set of
restrictions on the way in which the counters are accessed. By the end of this
section, we will have formalised these restrictions as \emph{counter-stack}
automata.

Suppose that we have a bounded-one counter automaton with counter $c$ and bound
$b = 15$. Hence, the width of the counter is~$4$ bits. Now suppose that we wish
to store two $2$-bit counters $c_1$ and $c_2$ in $c$. We can do this as follows:

\begin{center}
\begin{tikzpicture}[shorten >=1pt,node distance=1cm,auto]
\node[font=\huge] (c) {\texttt{c}};
\node[font=\huge] (equals) [right of=c] {\texttt{=}};
\node[font=\huge] (1) [right of=equals] {\texttt{1}};
\node[font=\huge] (2) [right of=1] {\texttt{0}};
\node[font=\huge] (3) [right of=2] {\texttt{0}};
\node[font=\huge] (4) [right of=3] {\texttt{1}};
\draw [dashed,rounded corners] (1.6, -0.5) rectangle (3.4, 0.5);
\draw [dashed,rounded corners] (3.6, -0.5) rectangle (5.4, 0.5);
\node[font=\Large] (l) at (2.6, -0.8) {$c_2$};
\node[font=\Large] (l) at (4.5, -0.8) {$c_1$};
\end{tikzpicture} 
\end{center}
We allocate the top two bits of~$c$ to store~$c_2$, and the bottom two bits to
store~$c_1$. We can easily write to both counters: if we want to increment~$c_2$
then we add~$4$ to~$c$, and if we want to increment~$c_1$ then we add~$1$
to~$c$.

However, if we want to test equality, then things become more interesting. It is
easy to test equality against $c_2$: if we want to test whether $c_2 = 2$, then
we test whether $8 \le c \le 11$ holds. But, we cannot easily test whether $c_1
= 2$ because we would have to test whether $c$ is $2$, $6$, $10$, or $14$, and
this list grows exponentially as the counters get wider. However, if we know
that $c_2 = 1$, then we only need to test whether $c = 6$. Thus, we arrive at
the following guiding principal: if you want to test equality against $c_i$,
then you must know the values of $c_j$ for all $j > i$. Counter-stack automata
are a formalisation of this principal. 

\paragraph{\bf Counter-stack automata.}

A counter-stack automaton has a set of~$k$ distinct counters, which are
referred to as~$c_1$ through~$c_k$. For our initial definitions, we will allow
the counters to take all values from $\nats$, but we will later refine this by
defining \emph{bounded} counter-stack automata. The defining feature of a
counter-stack automaton is that the counters are arranged in a stack-like
fashion:
\begin{itemize}
\item All counters may be increased at any time.
\item $c_i$ may only be tested for equality if the values of $c_{i+1}$ through $c_k$ are known.
\item $c_i$ may only be reset if the values of $c_i$ through $c_k$ are
known.
\end{itemize}

When the automaton increases a counter, it adds a specified number $n \in
\nats$ to that counter. The automaton has the ability to perform equality tests
against a counter, but the stack-based restrictions must be respected. An
example of a valid equality test would be $c_k = 3 \land c_{k-1} = 10$, because
$c_{k-1} = 10$ only needs to be tested in the case where $c_k = 3$ is known to
hold. Conversely, the test $c_{k-1} = 10$ by itself is invalid, because it
places no restrictions on the value of $c_k$. 

The automaton may also reset a counter, but the stack-based restrictions
apply. Counter~$c_i$ may only be reset by a transition, if that transition
tests equality against the values of $c_{i}$ through
$c_{k}$. For example, $c_{k-1}$ may only be reset if the transition is guarded
by a test of the form $c_{k-1} = n_1 \land c_{k-2} = n_2$.

\paragraph{\bf Formal definition.} 

A counter-stack automaton is a tuple $(L, C, \Delta, l_0)$, where $L$ is a
finite set of locations, $C = [1, k]$ is a set of counter indexes, $l_0 \in L$
is an initial state, and $\Delta$ specifies the transition relation. Each
transition in $\Delta$ has the form $(l, E, I, R, l')$ where:
\begin{itemize}
\item $l, l' \in L$ is a pair of locations,
\item $E$ is a partial function from $C$ to $\nats$ which specifies the equality
tests. If $E(i)$ is defined for some $i$, then $E(j)$ must be defined for all $j
\in C$ with $j > i$.
\item $I \in \nats^k$ specifies the how the counters must be increased,
\item $R \subseteq C$ specifies the set of counters that must be reset. It is
required that $E(r)$ is defined for every $r \in R$.
\end{itemize}

Each state of the automaton is a location annotated with values for each of
the~$k$ counters. That is, the state space of the automaton is $L \times
\nats^k$. A state $(l, c_1, c_2, \dots, c_k)$ can transition to a state $(l',
c_1', c_2', \dots, c_k')$ if, and only if, there exists a transition $(l, E, I,
R, l') \in \Delta$, where the following conditions hold:
\begin{itemize}
\item For every $i$ for which $E(i)$ is defined, we must have $c_i = E(i)$.
\item For every $i \in R$, we must have $c'_i = 0$.
\item For every $i \notin R$, we must have $c'_i = c_i + I_i$.
\end{itemize}

A \emph{run} is a sequence of states $s_0, s_1, \dots, s_n$, where each $s_i$
can transition to $s_{i+1}$. To solve the reachability problem for counter-stack
automata, we must decide whether there is a run from $(l_0, 0, 0, \dots, 0)$ to
a target state $(l_t, t_1, t_2, \dots, t_k)$.

A counter-stack automaton is \emph{$b$-bounded}, for some $b \in \nats$, if it
is impossible for the automaton to increase a counter beyond $b$. Formally, this
condition requires that, for every state $(l, c_1, c_2, \dots, c_k)$ that can be
reached by a run from $(l_0, 0, 0, \dots, 0)$, we have $c_i \le b$ for all $i$.
We say that a counter-stack automaton is bounded, if it is $b$-bounded for some
$b \in \nats$.

\paragraph{\bf Simulation by a bounded one-counter automaton.}
A bounded counter-stack automaton is designed to be simulated by a
bounded one-counter automaton. To do this, we follow the construction outlined
at the start of this section: we split the bits of the counter $c$ into $k$
chunks, where each chunk represents one of the counters~$c_i$. Note that the
boundedness assumption is crucial, because otherwise incrementing $c_i$ may
overflow the allotted space, and inadvertently modify the value of $c_{i+1}$.
See Appendix~\ref{app:bcsa2boca} for more details of the construction.

\begin{lemma}
\label{lem:bcsa2boca}
Reachability in bounded counter-stack automata is polynomial-time
reducible to reachability in bounded one-counter automata.
\end{lemma}

\section{Outline Of The Construction}

Our goal is to show that reachability in bounded counter-stack
automata is PSPACE-hard. To do this, we will show that subset-sum games can be
solved by bounded counter-stack automata. In this section, we give an
overview of our construction using the following two-round QSS game.
\begin{equation*}
\bigl ( \; \forall \; \{A_1, B_1\} \; \exists \; \{E_1, F_1\} \; \forall \;
\{A_2, B_2\} \; \exists \; \{E_2, F_2\}, \; T \bigr).
\end{equation*}
For brevity, we will refer to this instance as $(\psi, T)$ for the rest of this
section. The construction is split into two parts: the \emph{play} gadget, and
the \emph{reset} gadget.

\paragraph{\bf The play gadget.}

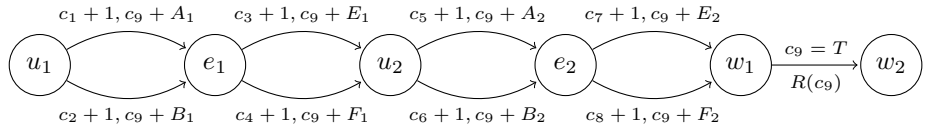
\begin{figure}
\begin{tikzpicture}[shorten >=1pt,node distance=2.33cm,auto,font=\scriptsize]
\node[state,font=\normalsize] (x) [] {$u_1$};
\node[state,font=\normalsize] (1) [right of=x] {$e_1$};
\node[state,font=\normalsize] (2) [right of=1] {$u_2$};
\node[state,font=\normalsize] (25) [right of=2] {$e_2$};
\node[state,font=\normalsize] (3) [right of=25] {$w_1$};
\node[state,font=\normalsize] (4) [right of=3,node distance=2cm] {$w_2$};
\path[->] 
	(x) edge [bend left] node [sloped,anchor=south] {$c_1 + 1, c_9 + A_1$} (1)
	(x) edge [bend right] node [sloped,anchor=north] {$c_2 + 1, c_9 + B_1$} (1)
	(1) edge [bend left] node [sloped,anchor=south] {$c_3 + 1, c_9 + E_1$} (2)
	(1) edge [bend right] node [sloped,anchor=north] {$c_4 + 1, c_9 + F_1$} (2)
	(2) edge [bend left] node [sloped,anchor=south] {$c_5 + 1, c_9 + A_2$} (25)
	(2) edge [bend right] node [sloped,anchor=north] {$c_6 + 1, c_9 + B_2$} (25)
	(25) edge [bend left] node [sloped,anchor=south] {$c_7 + 1, c_9 + E_2$} (3)
	(25) edge [bend right] node [sloped,anchor=north] {$c_8 + 1, c_9 + F_2$} (3)
	(3) edge node [anchor=south] {$c_9 = T$} node [anchor=north] {$R(c_9)$} (4)
	;
\end{tikzpicture} 
\caption{The play gadget}
\label{fig:guess}
\end{figure}

The play gadget is shown in Figure~\ref{fig:guess}. The construction uses $9$
counters. The locations are represented by circles, and the transitions are
represented by edges. The annotations on the transitions describe the
increments, resets, and equality tests: the notation $c_i + n$ indicates that
$n$ is added to counter $i$, the notation $R(c_i)$ indicates that counter~$i$ is
reset to~$0$, and the notation~$c_i = n$ indicates that the transition may only
be taken when~$c_i = n$ is satisfied.

This gadget allows the automaton to implement a play of the SSG. The
locations~$u_1$ and~$u_2$ allow the automaton to choose the first and second
moves of the universal player, while the locations~$e_1$ and~$e_2$ allow the
automaton to choose the first and second moves for the existential player. As
the play is constructed, a running total is stored in $c_9$, which
is the top counter on the stack. The final transition between~$w_1$ and~$w_2$
checks whether the existential player wins the play, and then resets $c_9$.
Thus, the set of runs between $u_1$ and $w_2$ corresponds precisely to the set
of plays won by the existential player in the SSG.

In addition to this, each outgoing transition from $u_i$ or $e_i$ comes
equipped with its own counter. This counter is incremented if and only if the
corresponding edge is used during the play, and this allows us to check
precisely which play was chosen. These counters will be used by the reset
gadget. The idea behind our construction is to force the automaton to pass
through the play gadget multiple times. Each time we pass through the play
gadget, we will check a different play, and our goal is to check a set of plays
that verify whether the existential player has a winning strategy for the SSG.

\paragraph{\bf Which plays should be checked?}
In our example, we must check four plays. The format of these plays is shown in
Table~\ref{tbl:plays}.
\vskip -0.5cm
\begin{table}
\begin{center}
\begin{tabular}{l | c | c | c | c}
Play \; & \; $u_1$ \; & \; $e_1$ \; & \; $u_2$ \; & \; $e_2$ \\ \hline
1 & $A_1$ & $E_1$ or $F_1$ & $A_2$ & \; $E_2$ or $F_2$ \; \\
2 & $A_1$ & \; Unchanged \; & $B_2$ & $E_2$ or $F_2$ \\
3 & $B_1$ & $E_1$ or $F_1$ & $A_2$ & $E_2$ or $F_2$ \\
4 & $B_1$ & Unchanged &  $B_2$ & $E_2$ or $F_2$ 
\end{tabular}
\end{center}
\caption{The set of plays that the automaton will check.}
\label{tbl:plays}
\end{table}

\vskip -0.7cm
The table shows four different plays, which cover every possible strategy choice
of the universal player. Clearly, if the existential player does have a winning
strategy, then that strategy should be able to win against all strategy choices
of the universal player. The plays are given in a very particular order: the
first two plays contain~$A_1$, while the second two plays contain~$B_1$.
Moreover, we always check~$A_2$, before moving on to~$B_2$.

We want to force the decisions made at $e_1$ and $e_2$ to form a coherent
strategy for the existential player. In this game, a strategy for the
existential player is a pair $\s = (s_1, s_2)$, where $s_i$ describes the move
that should be made at $e_i$. It is critical to note that $s_1$ only knows
whether $A_1$ or $B_1$ was chosen at $u_1$. This restriction is shown in the
table: the automaton may choose freely between $E_1$ and $F_1$ in the first
play. However, in the second play, the automaton must make the same choice as it
did in the first play. The same relationship holds between the third and fourth
plays. These restrictions ensure that the plays shown in Table~\ref{tbl:plays}
are a description of a strategy for the existential player.

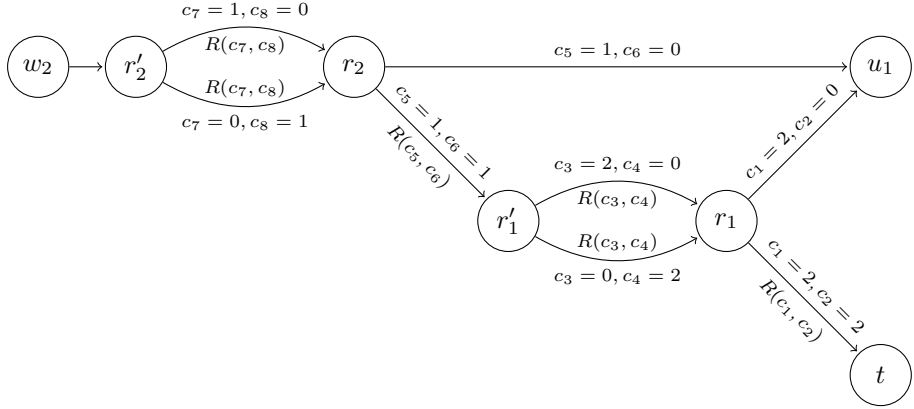
\begin{figure}
\begin{center}
\begin{tikzpicture}[shorten >=1pt,node distance=2.9cm,auto,font=\scriptsize]
\node[state,font=\normalsize] (w) {$w_2$};
\node[state,font=\normalsize] (r1)[right of=w,node distance=1.3cm] {$r'_2$};
\node[state,font=\normalsize] (t) [right of=r1]{$r_2$};
\node[state,font=\normalsize] (1) [below right of=t] {$r'_1$};
\node[state,font=\normalsize] (2) [right of=1] {$r_1$};
\node[state,font=\normalsize] (s) [above right of=2] {$u_1$};
\node[state,font=\normalsize] (g) [below right of=2] {$t$};
\path[->] 
	(w) edge (r1)
	(r1) edge [bend left] node [sloped,anchor=south] {$c_7 = 1, c_8 = 0$} node [sloped,anchor=north] {$R(c_7, c_8)$} (t)
	(r1) edge [bend right] node [sloped,anchor=north] {$c_7 = 0, c_8 = 1$} node
[sloped,anchor=south] {$R(c_7, c_8)$} (t)
	(t) edge node [sloped,anchor=south] {$c_5 = 1, c_6 = 0$} (s)
	(t) edge node [sloped,anchor=south] {$c_5 = 1, c_6 = 1$} node [sloped,anchor=north] {$R(c_5, c_6)$} (1)
	(1) edge [bend left] node [sloped,anchor=south] {$c_3 = 2, c_4 = 0$} node
[sloped,anchor=north] {$R(c_3, c_4)$} (2)
	(1) edge [bend right] node [sloped,anchor=north] {$c_3 = 0, c_4 = 2$} node
[sloped,anchor=south] {$R(c_3, c_4)$} (2)
	(2) edge node [sloped,anchor=south] {$c_1 = 2, c_2 = 0$} (s)
	(2) edge node [sloped,anchor=south] {$c_1 = 2, c_2 = 2$} node
[sloped,anchor=north] {$R(c_1, c_2)$} (g)
	;
\end{tikzpicture} 
\end{center}
\caption{The reset gadget}
\label{fig:loop}
\end{figure}

\paragraph{\bf The reset gadget.}
The reset gadget, shown in Figure~\ref{fig:loop}, enforces the constraints shown
in Table~\ref{tbl:plays}. The locations~$w_2$ and~$u_1$ represent the same
locations as they did in Figure~\ref{fig:guess}. To simplify the diagram, we
have only included meaningful equality tests. Whenever we omit a required
equality test, it should be assumed that the counter is~$0$. For example, the
outgoing transitions from~$r_2$ implicitly include the requirement
that~$c_7$,~$c_8,$ and~$c_9$ are all~$0$.

We consider the following reachability problem: can $(t, 0, 0, \dots, 0)$ be
reached from $(u_1, 0, 0, \dots, 0)$? The structure of the reset gadget places
restrictions on the runs that reach $t$. All such runs pass through the reset
gadget exactly four times, and the following table describes each pass:
\begin{center}
\setlength{\tabcolsep}{10pt}
\begin{tabular}{c | l}
Pass & Path \\\hline
1 & $w_2 \rightarrow r'_2 \rightarrow r_2 \rightarrow u_1$ \\
2 & $w_2 \rightarrow r'_2 \rightarrow r_2 \rightarrow r'_1 \rightarrow
r_1 \rightarrow u_1$ \\

3 & $w_2 \rightarrow r'_2 \rightarrow r_2 \rightarrow u_1$ \\

4 & $w_2 \rightarrow r'_2 \rightarrow r_2 \rightarrow r'_1
\rightarrow r_1 \rightarrow t$ \\
\end{tabular}
\end{center}
To see why these paths must be taken, observe that, for every~$i \in \{1, 3\}$,
each pass through the play gadget increments either~$c_i$
or~$c_{i+1}$, but not both. This means that the first time that we arrive
at~$r_2$, we must take the transition directly to~$u_1$, because the guard on
the transition to~$r'_1$ cannot possibly be satisfied after a single pass
through the play gadget. When we arrive at $r_2$ on the second pass, we are
forced to take the transition to $r'_1$, because we cannot have $c_5 = 1$ and
$c_6 = 0$ after two passes through the play gadget. This transition resets both
$c_5$ and $c_6$, so the pattern can repeat again on the third and fourth visits
to $r_2$. The location $r_1$ behaves in the same way as $r_2$, but the equality
tests are scaled up, because $r_1$ is only visited on every second pass through
the reset gadget.

We can now see that all strategies of the universal player must be considered.
The transition between~$r_2$ and~$u_1$ forces the play
gadget to increment~$c_5$, and therefore the first and third plays must
include~$A_2$. Similarly, the transition between~$r_2$ and~$r'_1$ forces the
second and fourth plays to include~$B_2$. Meanwhile, the transition between
$r_1$ and $u_1$ forces the first and second plays to include $A_1$, and the
transition between $r_1$ and $t$ forces the third and fourth plays to include
$B_1$. Thus, we select the universal player strategies exactly as
Table~\ref{tbl:plays} prescribes.

The transitions between~$r'_1$ and~$r_1$ check that the existential player is
playing a coherent strategy. When the automaton arrives at~$r'_1$ during the
second pass, it verifies that either~$E_1$ was included in the first and second
plays, or that~$F_1$ was included in the first and second plays. If this is not
the case, then the automaton gets stuck. The counters~$c_3$ and~$c_4$ are reset
when moving to~$r_1$, which allows the same check to occur during the fourth
pass. For the sake of completeness, we have included the transitions between
$r'_2$ and $r_2$, which perform the same check for $E_2$ and $F_2$. However,
since the existential player is allowed to change this decision on every pass,
the automaton can never get stuck at $r'_2$.

The end result is that location~$t$ can be reached if and only if the
existential player has a winning strategy for $(\psi, T)$. As we will show in
the next section, the construction extends to arbitrarily large SSGs,
which then leads to a proof that reachability in counter-stack automata is
PSPACE-hard. Note that all counters in this construction are bounded: $c_9$ is
clearly bounded by the maximum value that can be achieved by a play of the SSG,
and reset gadget ensures that no other counter may exceed~$4$. Thus, we will
have completed our proof of PSPACE-hardness for bounded one-counter automata and
two-clock timed automata.

\section{Formal Definition and Proof}

\paragraph{\bf Sequential strategies for SSGs.}

We start by formalising the ideas behind Table~\ref{tbl:plays}. Recall that the
table gives a strategy for the existential player in the form of a list of
plays. Moreover, the table gave a very specific ordering in which these plays
must appear. We now formalise this ordering.

We start by dividing the integers in the interval $[1, 2^n]$ into
\emph{$i$-blocks}. The $1$-blocks partition the interval into two equally sized
blocks. The first $1$-block consists of the range $[1, 2^{n-1}]$, and the second
$1$-block consists of the range $[2^{n-1} + 1, 2^n]$. There are four $2$-blocks,
which partition the $1$-blocks into two equally sized sub-ranges. This pattern
continues until we reach the $n$-blocks. 

Formally, for each $i \in \{1, 2, \dots, n\}$, then there are $2^{i}$ distinct
$i$-blocks. The set of $i$-blocks can be generated by considering the intervals
$[k + 1, k + 2^{n - i}]$ for the first~$2^i$ numbers $k \ge 0$ that satisfy $k
\bmod 2^{n - i} = 0$. An $i$-block is \emph{even} if $k$ is an even
multiple of $2^{n-i}$, and it is \emph{odd} if $k$ is an odd
multiple of $2^{n-i}$.

The ordering of the plays in Table~\ref{tbl:plays} can be described using
blocks. There are four $2$-blocks, and $A_2$ appears only in even $2$-blocks,
while $B_2$ only appears in odd $2$-blocks. Similarly, $A_1$ only appears in the
even $1$-block, while $B_1$ only appears in the odd $1$-block. The restrictions
on the existential player can also be described using blocks: the existential
player's strategy may not change between $E_i$ and $F_i$ during a $i$-block. We
generalise this idea in the following definition.

\begin{definition}[Sequential strategy]
\label{def:ss}
A sequential strategy for the existential player in $(\psi, T)$ is a list
of $2^n$ plays $\mathcal{S} = P_1, P_2, \dots, P_{2^n}$, where for every
$i$-block $L$ we have:
\begin{itemize}
\item If $L$ is an even $i$-block, then $P_j$ must contain $A_i$ for all $j \in
L$.
\item If $L$ is an odd $i$-block, then $P_j$ must contain $B_i$ for all $j \in
L$.
\item We either have $E_i \in P_j$ for all $j \in L$, or we have $F_i \in P_j$
for all $j \in L$.
\end{itemize}
\end{definition}

We say that $\mathcal{S}$ is winning for the existential player if $\sum P_j =
T$ for every $P_j \in \mathcal{S}$. Since a sequential strategy is simply a
strategy written in the form of a list, we have the following lemma. See
Appendix~\ref{app:ss} for further details.

\begin{lemma}
\label{lem:ss}
The existential player has a winning strategy if and only if the
existential player has a sequential winning strategy.
\end{lemma}

\paragraph{\bf The base automaton.} 

We describe the construction in two steps. Recall, from
Figures~\ref{fig:guess} and~\ref{fig:loop}, that the top counter is used by the
play gadget to store the value of the play, and to test whether the play is
winning. We begin by constructing a version of the automaton that omits the top
counter. That is, if $c_k$ is the top counter,  we modify the play gadget by
removing all increases to~$c_k$, and the equality test for~$c_k$ between~$w_1$
and~$w_2$. We call this the \emph{base} automaton. Later, we will add the
constraints for $c_k$ back in, to construct the \emph{full} automaton.

We now give a formal definition of the base automaton. Throughout this
definition, we keep consistency with the location and counter names used in
Figures~\ref{fig:guess} and~\ref{fig:loop}. For each natural number $n$, we
define a counter-stack automaton $\A_n$ as follows. The automaton has the
following set of locations
\begin{itemize}
\item For each $i \in [1, n]$ we have a location $u_i$ and a location~$e_i$.
\item We have two check states $w_1$ and $w_2$.
\item For each $i \in [1, n]$ we have two reset locations $r_i$ and $r'_i$.
\item We have a goal location $t$.
\end{itemize}
The automaton uses $k = 2n + 1$ counters. The top counter $c_k$ is reserved for
the full automaton, and will not be used in this construction. We will identify
counters $1$ through $2n$ using the following shorthands. For each integer $i$,
we define $a_i = c_{4(i-1) + 1}$, we define $b_i = c_{4(i-1) + 2}$, we define
$e_i = c_{4(i-1) + 3}$, and we define $f_i = c_{4(i-1) + 4}$. Note that, in
Figure~\ref{fig:guess}, we have $a_1 = c_1$ and $a_2 = c_5$, and these are
precisely the counters associated with $A_1$ and $A_2$, respectively. The same
relationship holds between $b_i$ and $B_i$, and so on.

The transitions of the automaton are defined as follows. Whenever we omit a
required equality test against a counter $c_i$, it should be assumed that the
transition includes the test $c_i = 0$.
\begin{itemize}
\item Each location $u_i$ has two transitions to $e_{i}$.
\begin{itemize}
\item A transition that adds $1$ to $a_i$.
\item A transition that adds $1$ to $b_i$.
\end{itemize}
\item We define $u_{n+1}$ to be a shorthand for $w_1$. Each location $e_i$ has
two transitions to $u_{i+1}$.
\begin{itemize}
\item A transition that adds $1$ to $e_i$.
\item A transition that adds $1$ to $f_i$.
\end{itemize}
\item Location $w_1$ has a transition to $w_2$, and $w_2$ has a transition to
$r'_n$. These transitions do not increase any counter, and do not test any
equalities.
\item Each location $r'_i$ has two outgoing transitions to $r_i$.
\begin{itemize}
\item A transition that tests $e_i = 2^{n-i}$ and $f_i = 0$.
\item A transition that tests $e_i = 0$ and $f_i = 2^{n-i}$.
\end{itemize}
\item We define $r'_{0}$ to be shorthand for location $t$. Each location $r_i$
has two outgoing transitions.
\begin{itemize}
\item A transition to $u_1$ that tests $a_i = 2^{n-i}$ and $b_i = 0$.
\item A transition to $r'_{i-1}$ that tests $a_i = 2^{n-i}$ and $b_i = 2^{n-i}$.
\end{itemize}
\end{itemize}

\paragraph{\bf Runs in the base automaton.} 

We now describe the set of runs are possible in the base automaton. We decompose
every run of the automaton into segments, such that each segment contains a
single pass through the play gadget. More formally, we decompose $R$ into
segments $R_1$, $R_2$, $\dots$, where each segment $R_i$ starts at $u_1$, and
ends at the next visit to $u_1$. We say that a run gets \emph{stuck} if the run
does not end at~$(t,0,0,\dots, 0)$, and if the final state of the run has
no outgoing transitions. We say that a run $R$ gets stuck during an $i$-block
$L$ if there exists a $j \in L$ such that $R_j$ gets stuck. The following lemma
gives a characterisation of the runs in $\A_n$. See Appendix~\ref{sec:ssruns}
for further details.

\begin{lemma}
\label{lem:ssruns}
Let $R$ be a run in $\A_n$. $R$ does not get stuck if and only if, for every $i$-block $L$, all of the following hold.
\begin{itemize}
\item If $L$ is an even $i$-block, then $R_j$ must increment $a_i$ for every $j
\in L$.
\item If $L$ is an odd $i$-block, then $R_j$ must increment $b_i$ for every $j
\in L$.
\item Either $R_j$ increments $e_i$ for every $j \in L$, or $R_j$ increments $f_i$ for every $j \in L$.
\end{itemize}
\end{lemma}

We say that a run is \emph{successful} if it eventually reaches $(t, 0, 0,
\dots, 0)$. By definition, a run is successful if and only if it never gets
stuck. Also, the transition from $r_1$ to $t$ ensures that every successful run
must have exactly $2^n$ segments. With these facts in mind, if we compare
Lemma~\ref{lem:ssruns} with Definition~\ref{def:ss}, then we can see that the
set of successful runs in $\A_n$ corresponds exactly to the set of sequential
strategies for the existential player in the SSG.

Since we eventually want to implement $\A_n$ as a bounded one-counter automaton,
it is important to prove the $\A_n$ is bounded. We do this in the following
Lemma. See Appendix~\ref{sec:bounded} for full details.

\begin{lemma}
\label{lem:bounded}
Along every run of $\A_n$ we have that:
\begin{itemize}
\item $a_i$ and $b_i$ are bounded by $2^{n-i+1}$, and
\item $e_i$ and $f_i$ are bounded by $2^{n-i}$.
\end{itemize}
\end{lemma}

\paragraph{\bf The full automaton.}

Let $(\psi, T)$ be an SSG instance, where $\psi$ is:
\begin{equation*}
\forall \; \{A_1, B_1\} \; \exists \; \{E_1, F_1\} \; \dots \; \forall \; \{A_n,
B_n\} \; \exists \; \{E_n, F_n\}.
\end{equation*} 
We will construct a counter-stack automaton $\A_\psi$ from $\A_n$. Recall that
the top counter $c_k$ is unused in $\A_n$. We modify the transitions of $\A_n$
as follows. Let $\delta$ be a transition. If $\delta$ increments $a_i$ then it
also adds $A_i$ to $c_k$, if $\delta$ increments $b_i$ then it also adds $B_i$
to $c_k$, if $\delta$ increments $e_i$ then it also adds $E_i$ to $c_k$, and if
$\delta$ increments $f_i$ then it also adds $F_i$ to $c_k$. We also modify the
transition between $w_1$ and $w_2$, so that it checks whether $c_k = T$, and
resets $c_k$.

Since we only add extra constraints to $\A_n$, the set of successful runs in
$\A_\psi$ is contained in the set of successful runs of $\A_n$. Recall that the
set of successful runs in $\A_n$ encodes the set of sequential strategies for
the existential player in $(\psi, T)$. In $\A_\psi$, we simply check whether
each play in the sequential strategy is winning for the existential player.
Thus, we have shown the following lemma.

\begin{lemma}
The set of successful runs in $\A_\psi$ corresponds precisely to the set of
winning sequential strategies for the existential player in $(\psi, T)$.
\end{lemma}

We also have that $\A_\psi$ is bounded. Counters $c_1$ through $c_{k-1}$ are
bounded due to Lemma~\ref{lem:bounded}, and counter~$c_k$ is bounded by $\sum
\{A_i, B_i, E_i, F_i \; : \; 1 \le i \le n\}$. This completes the reduction from
subset-sum games to bounded counter-stack automata, and gives us our main
theorem. 

\begin{theorem}
Reachability in bounded counter-stack automata is PSPACE-hard.
\end{theorem}

\begin{corollary}
We have:
\begin{itemize}
\item Reachability in bounded one-counter automata is PSPACE-complete.
\item Reachability in $2$-clock timed automata is PSPACE-complete.
\end{itemize}
\end{corollary}

\bibliographystyle{plain}
\bibliography{references}

\begin{thebibliography}{1}

\bibitem{AD94}
R.~Alur and D.~L. Dill.
\newblock A theory of timed automata.
\newblock {\em Theoretical Computer Science}, 126(2):183--235, 1994.

\bibitem{PFGMS08}
P.~Bouyer, U.~Fahrenberg, K.~G. Larsen, N.~Markey, and J.~Srba.
\newblock Infinite runs in weighted timed automata with energy constraints.
\newblock In {\em Proc.\ of FORMATS}, pages 33--47, 2008.

\bibitem{CLRS09}
T.~H. Cormen, C.~E. Leiserson, R.~L. Rivest, and C.~Stein.
\newblock {\em Introduction to Algorithms, Third Edition}.
\newblock The MIT Press, 2009.

\bibitem{CY92}
C.~Courcoubetis and M.~Yannakakis.
\newblock Minimum and maximum delay problems in real-time systems.
\newblock {\em Formal Methods in System Design}, 1(4):385--415, 1992.

\bibitem{HKOW09}
C.~Haase, S.~Kreutzer, J.~Ouaknine, and J.~Worrell.
\newblock Reachability in succinct and parametric one-counter automata.
\newblock In {\em Proc.\ of CONCUR}, pages 369--383, 2009.

\bibitem{HOW12}
C.~Haase, J.~Ouaknine, and J.~Worrell.
\newblock On the relationship between reachability problems in timed and
  counter automata.
\newblock In {\em Proc.\ of RP}, pages 54--65, 2012.

\bibitem{LMS04}
F.~Laroussinie, N.~Markey, and P.~Schnoebelen.
\newblock Model checking timed automata with one or two clocks.
\newblock In {\em Proc.\ of CONCUR}, pages 387--401, 2004.

\bibitem{Naves}
G.~Naves.
\newblock Accessibilit\'{e} dans les automates temporis\'{e} \`{a} deux
  horloges.
\newblock Rapport de {M}aster, MPRI, Paris, France, 2006.

\bibitem{T11}
S.~Travers.
\newblock The complexity of membership problems for circuits over sets of
  integers.
\newblock {\em Theoretical Computer Science}, 369(1–3):211--229, 2006.

\end{thebibliography}

\appendix 
\newpage

\section{Proof of Lemma~\ref{lem:qss}}
\label{app:qss}

\paragraph{\bf Outline.} 
A quantified version of subset-sum has already been shown to be
PSPACE-hard~\cite{T11}, and the proof easily carries over for the case of SSGs.
For the sake of completeness, we provide a direct proof that SSGs are
PSPACE-hard, which closely follows the ideas laid out in~\cite{T11}.

The proof follows the NP-hardness proof for subset-sum, taken from
\cite{CLRS09}[Theorem 34.10]. The key observation is that, if we begin with a
quantified version of 3-SAT, then we end up with an SSG.

\paragraph{\bf Subset-sum is NP-hard.} We now give a summary of the
NP-hardness proof given in \cite{CLRS09}[Theorem 34.10]. We will describe
the reduction using a worked example taken from~\cite{CLRS09}. Consider the
following 3-CNF formula:
\begin{align*}
\phi &= C_1 \land C_2 \land C_3 \land C_4 \\
C_1 &= (x_1 \lor \lnot x_2 \lor \lnot x_3) \\
C_2 &= (\lnot x_1 \lor \lnot x_2 \lor \lnot x_3) \\
C_3 &= (\lnot x_1 \lor \lnot x_2 \lor x_3) \\
C_4 &= (x_1 \lor x_2 \lor x_3)
\end{align*}
This formula has three variables, $x_1$, $x_2$, and $x_3$, and four clauses,
$C_1$ through $C_4$. The reduction assumes that there is no clause $C_i$ that
contains both $x_i$ and $\lnot x_i$, because otherwise $C_i$ would be always be
satisfied. 

The reduction constructs a subset-sum instance, which is described in
the following table:
\begin{center}
\begin{tabular}{rrrrrrrrr}
&& \; $x_1$ & \; $x_2$ & \; $x_3$ & \; $C_1$ & \; $C_2$ & \; $C_3$ & \; $C_4$ \\ \hline
$v_1$ & = & 1 & 0 & 0 & 1 & 0 & 0 & 1 \\
$v'_1$ & = & 1 & 0 & 0 & 0 & 1 & 1 & 0 \\ \hline
$v_2$ & = & 0 & 1 & 0 & 0 & 0 & 0 & 1 \\
$v'_2$ & = & 0 & 1 & 0 & 1 & 1 & 1 & 0 \\ \hline
$v_3$ & = & 0 & 0 & 1 & 0 & 0 & 1 & 1 \\
$v'_3$ & = & 0 & 0 & 1 & 1 & 1 & 0 & 0 \\ \hline
$s_1$ & = & 0 & 0 & 0 & 1 & 0 & 0 & 0 \\
$s'_1$ & = & 0 & 0 & 0 & 2 & 0 & 0 & 0 \\ 
$s_2$ & = & 0 & 0 & 0 & 0 & 1 & 0 & 0 \\
$s'_2$ & = & 0 & 0 & 0 & 0 & 2 & 0 & 0 \\ 
$s_3$ & = & 0 & 0 & 0 & 0 & 0 & 1 & 0 \\
$s'_3$ & = & 0 & 0 & 0 & 0 & 0 & 2 & 0 \\ 
$s_4$ & = & 0 & 0 & 0 & 0 & 0 & 0 & 1 \\
$s'_4$ & = & 0 & 0 & 0 & 0 & 0 & 0 & 2 \\ \hline
$t$ & = & 1 & 1 & 1 & 4 & 4 & 4 & 4 
\end{tabular}
\end{center}
Each row should be read as a number written in decimal. For example, the first
row specifies the number $v_1 = 1001001$. The subset-sum instance asks whether
there is a subset of rows~$v_1$ through~$s'_4$ that sums to row~$t$.

The table is constructed according to the following rules. Each column is
labelled: the first three columns are labelled by the variables~$x_1$
through~$x_3$, and the rest of the columns are labelled by the clauses~$C_1$
through~$C_4$.  For each variable $x_i$ we define two rows:
\begin{itemize}
\item $v_i$ has a~$1$ in column $x_i$, and a $1$ in every column $C_i$
that contains~$x_i$.
\item $v'_i$ has a~$1$ in column $x_i$, and a $1$ in every column $C_i$ that contains~$\lnot x_i$.
\end{itemize}
In addition to these, for each clause $C_i$ we define two \emph{slack} rows: the
row $s_i$ has a $1$ in column $C_i$, and the row $s'_i$ has a $2$ in column
$C_i$.

To see that this reduction works, suppose that we know a satisfying assignment
of the CNF formula. We can use this to construct a solution to the subset-sum
instance. If $x_i$ is true in the satisfying assignment, then we select $v_i$,
and if it is false then we select $v'_i$. In doing so, we construct a subset
with the following properties:
\begin{itemize}
\item For each column $x_i$, we have that the sum of that column is $1$, because
we never select both $v_i$ and $v'_i$.
\item For each column $C_i$, we have that the sum of that column is at least
$1$, because every clause must be satisfied.
\item For each column $C_i$, we have that the sum of that column is at most $3$,
because each clause contains exactly~$3$ variables.
\end{itemize}
These properties ensure that, for each column~$C_i$, we can always select a
subset of the slack columns,~$s_i$ and~$s'_i$, so that the sum of the column
is~$4$. Thus, every satisfying assignment of the CNF formula corresponds to a
solution of the subset-sum instance.

For similar reasons, every solution of the subset-sum instance corresponds to a
satisfying assignment of the CNF formula, by simply ignoring the slack rows.
Since every column $C_i$ must sum to $4$, we know that after removing the
slacks, each column must sum to at least $1$. This, combined with the fact that
$v_i$ and $v'_i$ cannot be selected at the same time, implies that we have a
satisfying assignment for the CNF formula.

See \cite{CLRS09} for a full proof correctness of the NP-hardness reduction.

\paragraph{\bf Changing the format.} 
Our definition of an SSG requires a very specific format for the input instance.
In particular, each quantifier is associated with exactly two natural numbers.
However, the reduction that we have described can be written down very naturally
as a one-player SSG, in which only the existential player is allowed to move.
For our example, the instance is $(V S_1 S_2 S_3 S_4, t)$, where:
\begin{align*} 
V &= \exists \; \{v_1, v'_1\} \; \exists \; \{v_2, v'_2\} \; \exists \; \{v_3,
v'_3\}, \\
S_i &= \exists \; \{s_i, 0\} \; \exists \; \{s'_i, 0\}.
\end{align*}
Note that it is valid to force the choice between $v_i$ and $v'_i$, because no
solution of the subset-sum instance can contain both of these numbers.

\paragraph{\bf Subset-sum games are PSPACE-complete.} It is now easy to reduce a
quantified boolean formula to an SSG. We simply follow the existing reduction,
but if variable $x_i$ is universally quantified, then we use $\forall \{v_i,
v'_i\}$ rather than $\exists \{v_i, v'_i\}$. For example, if we consider the
quantified boolean formula $\forall x_1 \exists x_2 \forall x_3 \; \phi$, where
$\phi$ is defined as before, then we produce the quantified subset-sum instance
$(V' S_1 S_2 S_3 S_4, t)$, where:
\begin{equation*}
V' = \forall \{v_1, v'_1\} \; \exists \{v_2, v'_2\} \; \forall \{v_3, v'_3\}, \\
\end{equation*}
and $S_i$ is defined as before.
The final step is to ensure a strict alternation of quantifiers, which the
definition of an SSG requires. This can easily be achieved by inserting
``dummy'' quantifiers, where necessary. That is, we can insert $\exists \{0,
0\}$ between two consecutive $\forall$ quantifiers, and we can insert $\forall
\{0, 0\}$ between two consecutive $\exists$ quantifiers. This change obviously
cannot affect the winner of the SSG.

\section{Proof of Lemma~\ref{lem:bcsa2boca}}
\label{app:bcsa2boca}

Let $\mathcal{S} = (L, C, \Delta, l_0)$ be a $b$-bounded counter-stack
automaton. Without loss of generality, we will assume that $b = 2^{n} - 1$,
which means that each counter in $S$ is $n$~bits wide. We will construct a
bounded one-counter automaton $\mathcal{B} = (L', b', \Delta', l'_0)$ that
simulates~$\mathcal{S}$. We will refer to the counters of $\mathcal{S}$ as $c_1$
through $c_k$, and the counter of~$\mathcal{B}$ as~$c$.

We will follow the approach laid out at the start of Section~\ref{sec:stack}.
That is, we will set the bound $b' = 2^{k \cdot n} - 1$ so that~$c$ is $k \cdot
n$ bits wide. We then partition these bits in order to implement the
counters~$c_1$ through~$c_k$. The counter $c_k$ will use the~$n$ most
significant bits, the counter $c_{k-1}$ will use the next~$n$ most significant
bits, and so on.

We introduce some notation to formalise this encoding. Let $x \in [0, b]$ be  a
counter value for counter $c_i$. We define $\enc(x, i) = x \cdot 2^{(i-1) \cdot
n}$. To understand this definition, note that for $i = 1$, we have $\enc(x, i) =
x$. Then, for $i = 2$, we have that $\enc(x, i)$ is the value of $x$ bit-shifted
to the left $n$~times. Thus, this definition simply translates $x$ to the
correct position in $c$.

We can now define the translation. We will set $L' = L$ and $l'_0 = l_0$, which
means that both automata have the same set of locations, and the same start
location. We will use the transitions in $\Delta'$ to simulate $S$. For each
transition $t = (l, E, I, R, l') \in \Delta$, we construct a transition
$t' = (l, p, g_1, g_2, l') \in \Delta'$ between the same pair of locations. We
want to have the following property: transition $t$ can be used from a state
$(l, c_1, c_2, \dots, c_k)$ in $\mathcal{S}$ if and only if transition $t'$ can
be used from the state $(l, \sum_{i} \enc(c_i, i))$ in $\mathcal{B}$.

We begin by defining $p$. We set:
\begin{equation*}
p = \sum_{i \notin R} \enc(I_i, i) - \sum_{i \in R} \enc(E(i), i).
\end{equation*}
In other words, for each counter $i \notin R$ that is not to be reset, we add
$\enc(I_i, i)$ to $c$, which correctly adds $I_i$ to $c_i$. Note that the
boundedness assumption on $\mathcal{S}$ implies that the counters can never
overflow due to this operation. For the counters $i \in R$, we subtract $E(i)$
from $c_i$. Recall that $E(i)$ must always be defined for the indices $i \in R$.
Furthermore, the transition may only be taken if~$c_i = E(i)$. Thus,
subtracting~$E(i)$ from~$c_i$ will correctly set it to~$0$.

Next we define the inequality tests. Let $j$ be the smallest index for which
$E(j)$ is defined. Our guards are:
\begin{align*}
g_1 &= \sum_{i \ge j} \enc(E(i), i), \\ 
g_2 &= \sum_{i \ge j} \enc(E(i), i) + \enc(1, j) - 1.
\end{align*}
It is straightforward to show that, in our encoding scheme, we have $c_i = E(i)$
for all $i \ge j$ if and only if $g_1 \le c \le g_2$.

If we are given a target state $s = (t, c_1, c_2, \dots, c_k)$ for
$\mathcal{S}$, then we can translate it into a target state $s' = (t, \sum_{i}
\enc(c_i, i))$ for $\mathcal{B}$. The equivalence between the transitions
in~$\Delta$, and the transitions in~$\Delta'$ implies that~$s$ can be reached
from $(l_0, 0, 0, \dots, 0)$ if and only if~$s'$ can be reached from $(l'_0,
0)$. This completes the proof of Lemma~\ref{lem:bcsa2boca}.

\section{Proof of Lemma~\ref{lem:ss}}
\label{app:ss}

Let $\s = (s_1, s_2, \dots, s_n)$ be a winning strategy for the existential
player. We define a sequential winning strategy as follows. Recall that
$\plays(\s)$ contains exactly $2^n$ plays. We argue that these plays can be
ordered so that they form a sequential strategy. We give an iterative procedure
that achieves this task: the first step of the procedure will ensure that the
$1$-blocks contain the correct plays, the second step will ensure that the
$2$-blocks contain the correct plays, and so on. In the first step, we observe
that exactly $2^{n-1}$ of the plays contain $A_1$, while exactly $2^{n-1}$ of
the plays contain $B_1$, so we can order the plays so that the even $1$-block
contains all plays containing $A_1$. Now suppose that we have found the
$i$-blocks. We observe that each $i$-block $L$ has exactly $2^{n-(i+1)}$ plays
that contain $A_{i+1}$. Therefore, for each $i$-block $L$, we can order the
plays in $L$ so that the even $(i+1)$-block has all plays that contain
$A_{i+1}$, and the odd $(i+1)$-block has all plays that contain $B_{i+1}$. At
the end of this procedure, we will have a list of plays $\mathcal{S} = P_1, P_2,
\dots, P_{2^n}$ where:
\begin{itemize}
\item $P_j$ contains $A_i$ whenever $j$ is in an even $i$-block.
\item $P_j$ contains $B_i$ whenever $j$ is in an odd $i$-block.
\end{itemize}
So $\mathcal{S}$ satisfies the first two conditions of Definition~\ref{def:ss}.
We argue that $\mathcal{S}$ also satisfies the third condition. Let $L$ be an
$i$-block. By definition, for every $j < i$, there is a unique $j$-block that
contains $L$. These blocks define a play prefix $F \in \Pi_{1 \le j \le i}
\{A_i, B_i\}$, and, for each play $P_j$ with $j \in L$, we have $F \subseteq
P_j$. Since $\mathcal{S}$ is a reordering of $\plays(\s)$, we must have $s_i(F)
\in P_j$ for every $j \in L$. Hence, $\mathcal{S}$ satisfies
Definition~\ref{def:ss}. Moreover, since $\s$ is winning, we have that every
play in $\plays(\s)$ is winning, and therefore $\mathcal{S}$ is a sequential
winning strategy.

Now let $\mathcal{S} = P_1, P_2, \dots, P_{2^n}$ be a winning sequential
strategy. We give a high level description of a winning strategy for the SSG. At
the start of the strategy we set $L_0 = [1, 2^n]$. In each round $i$ of the
game, let $D_i \in \{A_i, B_i\}$ be the decision made by the universal player.
We select $L_i$ to be the unique $i$-block in $L_{i-1}$ such that $D_i \in P_j$
for all $j \in L_{i}$. We play $E_i$ if $E_i \in P_j$ for all $j \in L_{i}$, and
we play $F_i$ if $F_i \in P_j$ for all $j \in L_i$. It is straightforward to
encode this strategy in the form $\s = (s_1, s_2, \dots, s_n)$. By construction,
when we play $\s$, the outcome of the game will be some play $P_j$ from
$\mathcal{S}$. Since every play $P_j$ in $\mathcal{S}$ is winning for the
existential player, we have that $\s$ is a winning strategy.

\section{Proof of Lemma~\ref{lem:ssruns}}
\label{sec:ssruns}

Let $R$ be a run in $\mathcal{A}_n$. The following lemma describes the set of
reset states that each segment of $R$ must pass through.

\begin{lemma}
\label{lem:runs}
Let $R$ be a run in $\A_n$. Either:
\begin{itemize}
\item $R_j$ visits precisely the reset locations $\{r'_i, r_i\}$ for which $j \bmod 2^{n-i} = 0$, or
\item $R_j$ gets stuck.
\end{itemize}
\end{lemma}
\begin{proof}
We will prove this lemma by induction over $i$. The base case, where $i = n$, is
trivial because $j \bmod 2^{n-n}$ is always equal to $0$, and it is clear from
the construction that every segment $R_j$ must always visit both $r'_n$ and
$r_n$.

For the inductive step, suppose that the lemma has been shown for $i + 1$, and
will show that the lemma holds for $i$. We know that, in order to reach $r'_i$
or $r_i$, a segment must first visit $r'_{i+1}$. By the inductive hypothesis, we
know that only segments $R_j$ with $j \bmod 2^{n - (i + 1)}$
visit~$r_{i+1}$. At the start of $R$, we have $a_i = b_i = 0$. On the first
visit to~$r_{i+1}$, we clearly cannot take the transition to $r'_i$, because we
have $a_i + b_i = 2^{n-(i+1)}$, and the transition to $r'_i$ requires $a_i + b_i
= 2^{n-i}$. Thus, we either have to take the transition to~$u_1$, or we get
stuck. On the second visit to~$r_{i+1}$, we cannot take the transition to~$u_1$,
because we have $a_i + b_i = 2^{n-i}$, and the transition to $u_1$ requires $a_i
+ b_i = 2^{n-(i+1)}$. Thus, either we get stuck, or we take the transition
to~$r'_{i}$. The transition between~$r_{i+1}$ and~$r'_i$ resets~$a_i$ and~$b_i$.
Thus, we can repeat the argument, and conclude that locations~$r'_i$ and~$r_i$
are only visited by segments~$R_j$ where $j \bmod 2^{n - i} = 0$.
\qed
\end{proof}

Having shown Lemma~\ref{lem:runs} it is now easy to prove
Lemma~\ref{lem:ssruns}. Let $R$ be a run of $\A_n$.
For the counters $a_i$ and $b_i$, we have the following facts:
\begin{itemize}
\item At the start of the first $i$-block, we have $a_i = b_i = 0$.
\item Each $i$-block contains exactly $2^{n-i}$ segments. Each segment must
increment one of $a_i$ or $b_i$, but not both.
\item At the end of each odd $i$-block, we must take the transition from $r_i$
to $u_1$ to avoid getting stuck. This transition requires $a_i = 2^{n-i}$ and $b_i = 0$.
\item At the end of each even $i$-block, we must take the transition from $r_i$
to $r'_{i-1}$ to avoid getting stuck. This transition requires $a_i = 2^{n-i}$
and $b_i = 2^{n-i}$, and resets $a_i$ and $b_i$ to $0$.
\end{itemize}
These facts imply that $a_i$ must be incremented during every run in an odd
$i$-block to prevent the automaton getting stuck, and $b_i$ must be incremented
during every run in an even $i$-block to prevent the automaton getting stuck. It
can also be verified that, if $a_i$ is incremented during every run in an odd
$i$-block, and $b_i$ is incremented during every run in an even $i$-block, then
the automaton will never get stuck at $r_i$.

Similarly, for the counters $e_i$ and $f_i$ we have the following facts.
\begin{itemize}
\item At the start of the first $i$-block, we have $e_i = f_i = 0$.
\item Each $i$-block contains exactly $2^{n-i}$ runs. Each run must increment
one of $e_i$ or $f_i$, but not both.
\item At the end of each $i$-block, we must take one of the two transitions from
$r'_i$ to $r_i$ to avoid getting stuck. These transitions require that $e_i =
2^{n-i}$ and $f_i = 0$, or $e_i = 0$ and $f_i = 2^{n-i}$.
\end{itemize}
These facts imply that either $e_i$ is incremented during every run in an
$i$-block, or $f_i$ is incremented during every run in an $i$-block, or the
automaton will get stuck when moving from $r'_i$ to $r_i$ at  end of the
$i$-block. It can also be verified that, if the automaton increases $e_i$ during
every run in an $i$-block, then the automaton will not get stuck moving from
$r'_i$ to $r_i$, and if the automaton increases $f_i$ during every run in an
$i$-block, then the automaton will not get stuck moving from $r'_i$ to $r_i$.

Note that, in $\A_n$, it is only possible for $R$ to get stuck at the locations
$r'_i$ and $r_i$. Therefore, we have shown that $R$ does not get stuck if and
only if the three conditions of Lemma~\ref{lem:ssruns} hold for $R$.

\section{Proof of Lemma~\ref{lem:bounded}}
\label{sec:bounded}

This lemma follows from Lemma~\ref{lem:runs}. Let $R$ be a run.
Lemma~\ref{lem:runs} implies that the transition from $r_i$ to $r'_{i-1}$ is
taken in every segment $R_j$ such that $j \bmod 2^{n-(i-1)}$. This transition
resets both $a_i$ and $b_i$ to $0$. Therefore, neither of these counters may
exceed $2^{n-(i-1)}$. Similarly, Lemma~\ref{lem:runs} implies that every segment
$R_j$ such that $j \bmod 2^{n-i} = 0$ must move from $r'_i$ to $r_i$. Both of
the transitions from between $r'_i$ and $r_i$ reset $e_i$ and $f_i$, and
therefore neither of these counters may exceed $2^{n-i}$.

\end{document}